\RequirePackage{silence}
\documentclass[aps,prx,10pt,groupedaddress,twocolumn,superscriptaddress]{revtex4-2}
\pdfoutput=1
\usepackage[T1]{fontenc} 
\usepackage[english]{babel} 
\usepackage[utf8]{inputenc}

\usepackage[table,usenames,dvipsnames]{xcolor}
\usepackage[section]{placeins}
\usepackage{amsmath,amssymb,mathtools,amsthm,dsfont}

\usepackage{paralist}
\usepackage{enumitem}

\usepackage[colorlinks=true,
			hypertexnames=false
			]{hyperref}
\PassOptionsToPackage{linktocpage}{hyperref} 
\usepackage{tikz}
\usepackage{tikz-cd}
\usetikzlibrary{positioning}

\usepackage[nolist,withpage]{acronym}
\makeatletter
\AtBeginDocument{%
  \renewcommand*{\AC@hyperlink}[2]{%
    \begingroup
      \hypersetup{hidelinks}%
      \hyperlink{#1}{#2}%
    \endgroup
  }%
}
\makeatother

\DeclareFontFamily{U}{mathb}{\hyphenchar\font45}
\DeclareFontShape{U}{mathb}{m}{n}{
      <5> <6> <7> <8> <9> <10> gen * mathb
      <10.95> mathb10 <12> <14.4> <17.28> <20.74> <24.88> mathb12
      }{}
\DeclareSymbolFont{mathb}{U}{mathb}{m}{n}
\DeclareMathSymbol{\Asterisk}{3}{mathb}{"06}
\usepackage{mathrsfs}
\DeclareMathAlphabet{\mathscrbf}{OMS}{mdugm}{b}{n} 
\usepackage{easybmat} 
\usepackage{cleveref}

\usepackage{etoolbox}

\renewcommand{\i}{\ensuremath\mathrm{i}} 
\DeclareMathOperator{\rank}{rank} 

\newcommand{\argdot}{{\,\cdot\,}} 

\DeclarePairedDelimiterX{\abs}[1]{\lvert}{\rvert}{%
  \ifblank{#1}{\,\cdot\,}{#1}
}   

\DeclarePairedDelimiterX\norm[1]\lVert\rVert{%
  \ifblank{#1}{\,\cdot\,}{#1}
}   

\providecommand\given{}
\newcommand\SetSymbol[1][]{%
  \nonscript\:#1\vert
  \allowbreak
  \nonscript\:
  \mathopen{}}
\renewcommand\SetSymbol[1][]{\,\colon} 
\DeclarePairedDelimiterX\Set[1]\{\}{%
  \renewcommand\given{\SetSymbol[\delimsize]}
  #1
}

\DeclarePairedDelimiterX\innerp[2]{\langle}{\rangle}{%
  \ifblank{#1}{\,\cdot\,}{#1} , \ifblank{#2}{\,\cdot\,}{#2}%
}
\DeclarePairedDelimiterX\scprod[2]{\langle}{\rangle}{%
  \ifblank{#1}{\,\cdot\,}{#1} , \ifblank{#2}{\,\cdot\,}{#2}%
}

\DeclarePairedDelimiterX\av[1]{\langle}{\rangle}{%
  \ifblank{#1}{\,\cdot\,}{#1}%
}

\DeclarePairedDelimiterX\braket[2]{\langle}{\rangle}%
  {#1\kern0.15ex\delimsize\vert\kern0.15ex\mathopen{}#2}

\DeclarePairedDelimiterX\ketbra[2]{\vert}{\vert}%
  {#1\kern0.15ex\delimsize\rangle\delimsize\langle\kern0.15ex\mathopen{}#2}

\DeclarePairedDelimiterX\sandwich[3]{\langle}{\rangle}%
  {#1\,\delimsize\vert\kern0.15ex\mathopen{}#2\kern0.15ex\delimsize\vert\kern0.15ex\mathopen{}#3}

\DeclareMathOperator*{\argmax}{arg\max}

\DeclareMathOperator{\conv}{\ast}
\DeclareMathOperator{\Conv}{\Asterisk}
\DeclareMathOperator{\supp}{supp} 

\newcommand{\pauliphases}{\mathcal{P}}
\newcommand{\pauli}{\mathsf{P}} 
\newcommand{\paulids}{\mathsf G^{n,m}} 
\newcommand{\Ftwo}{\mathbb{F}_2}
\newcommand{\stabs}{\mathscr{S}} 
\newcommand{\gauge}{\mathscr{G}} 
\newcommand{\logops}{\mathscr{L}} 
\newcommand{\meas}{\mathscr{M}} 
\newcommand{\undet}{\mathscr{U}} 
\newcommand{\fourier}[1]{\mathscr{F}[#1]} 
\newcommand{\fourierinv}[1]{\mathscr{F}^{-1}[#1]}
\newcommand{\kron}[1]{\left[#1\right]}

\newtheorem{theorem}{Theorem} 
\newtheorem{lemma}[theorem]{Lemma}

\newtheorem{definition}[theorem]{Definition} 
\newtheorem{corollary}[theorem]{Corollary} 

\newcommand{\hhu}{Institute for Theoretical Physics,
      Heinrich-Heine-University D{\"u}sseldorf,
      Germany
}

\newcommand{\tuhh}{%
    Institute for Quantum-Inspired and Quantum Optimization,
    Hamburg University of Technology, Germany
}

\begin{document}
\title{
Learning 
logical Pauli noise in quantum error correction
}

\author{Thomas Wagner}
\affiliation{\hhu}
\author{Hermann Kampermann}
\affiliation{\hhu}
\author{Dagmar Bru\ss}
\affiliation{\hhu}
\author{Martin Kliesch}
\email[]{martin.kliesch@tuhh.de}
\affiliation{\hhu}
\affiliation{\tuhh}

\hypersetup
 {pdftitle = {Learning Logical Pauli Noise in Quantum Error Correction},
       pdfauthor = {Thomas Wagner, Hermann Kampermann, Dagmar Bruss, Martin Kliesch},
       pdfsubject = {Quantum computing},
       pdfkeywords = {
                      quantum, error, correction, QEC, code, codes, QECC,
                      data, syndrome,
                      stabilizer, decoding,
                      pure, distance,
                      estimation, noise, from, syndrome, statistics,
                      probabilistic modeling, estimation, inverse problem,
                      measurements, identifiability,
                  Pauli, group, phase space, symplectic,
                      logical error, method of moments, moments,
                      binomial, system, equation,
                      Boolean, Fourier, transform,
                      orthogonal array,
                      Möbius, inclusion, exclusion, principle,
                      distributed source coding,
                      cleaning lemma
                      }
  }

\begin{abstract}
The characterization of quantum devices is crucial for their practical implementation but can be costly in experimental effort and classical postprocessing. 
Therefore, it is desirable to measure only the information that is relevant for specific applications and develop protocols that require little additional effort.
In this work, we focus on the characterization of quantum computers in the context of stabilizer quantum error correction.
For arbitrary stabilizer codes, subsystem codes, and data syndrome codes, we prove that the logical error channel induced by Pauli noise can be estimated from syndrome data under minimal conditions. 
More precisely, for any such code, we show that the estimation is possible as long as the code can correct the noise.
\end{abstract}

\maketitle

For any quantum device, it is desirable to characterize both its individual components as well as their interplay
\cite{kliesch2021_certificationreview, Eisert2020QuantumCertificationAnd}.
For the characterization of single quantum gates, protocols such as
quantum process tomography (e.g.\ Ref.~\cite{KliKueEis19})
or
gate set tomography \cite{BluGamNie17,Nielsen20pyGSTi,Brieger21CompressiveGateSet} can be used.
To characterize the interplay of multiple components, randomized benchmarking \cite{KniLeiRei08,Magesan2012}, as well as
crosstalk detection \cite{sarovar2020_detectingcrosstalk} and
estimation \cite{McKay2020CorrelatedRB,harper2020_efficientlearningofquantumnoise} protocols are available.
The general goals are
\begin{compactenum}[(i)]
\item to build trust in the correct functioning of the device,
\item to be able to reduce the errors on the hardware level and improve the software calibration, and
\item to compare different devices and platforms in a fair way.
\end{compactenum}
However, such characterization protocols can be quite resource-intensive, requiring many experimental runs of the device and  such protocols' output can be challenging to interpret.
Therefore, it has become a pressing issue to obtain easy-to-use information, such as Pauli error rates directly \cite{Flammia2019EfficientEstimation,flammia2021_paulilearningpopulationrecovery, Harper2020FastEstimationOf,harper2020_efficientlearningofquantumnoise,HarFlaWal19},
ideally using only data that is easy to obtain.
The estimation of Pauli noise is also practically interesting because randomized compiling can be used to project the actual noise onto Pauli noise \cite{Wallman16NoiseTailoringFor, Ware21ExperimentalPauli-frame}. This has explicitly been discussed in the context of \ac{QEC}  \cite{Iyer22EfficientDiagnosticsFor}.

In the context of \ac{QEC},
it has been suggested to reduce the experimental effort of characterization by extracting information from the syndrome data, which is usually collected during error correction anyway
\cite{fujiwara_instantaneouschannelestimationcss,
fowler_scalableextractionoferrormodelsfromqec,
Huo17LearningTime-dependent,
wootton_qiskitbenchmarking,
florjanczykbrun_insituadaptiveencoding,
obrien_adaptiveweightestimator,
combes_insitucharofdevice, Wagner2020OptimalNoiseEstimation, Wagner21PauliChannels}.
Such an approach is complementary to the standard benchmarking before operation.
It has the additional advantage of benchmarking all components in the context of the targeted application and making it easier to detect crosstalk.
Indeed, syndrome data has been used to calibrate decoders and observe signatures of crosstalk in experiments on the [4,1,2] code \cite{IBM2022_calibrateddecodersexperimental}, the repetition code \cite{google2021_repcode-short_auth_list} and the surface code \cite{google2023_dist5surfacecode_shortauthor}.
Finally, estimation based on syndrome data is the only method of characterization that is not invasive, in the sense that the encoded logical information is not perturbed by the measurements.
Thus, it is at least in principle suited for estimation of noise in a time-dependent environment \cite{obrien_adaptiveweightestimator,huo_2017}.

For general stabilizer codes, however, the theoretical foundation of such schemes is currently lacking.
Since the syndrome measurements must preserve the encoded state, it is not \emph{a priori} clear that they should even contain sufficient information about the noise to be useful for \ac{QEC}.
For example, as shown in our previous work \cite{Wagner21PauliChannels}, a complete Pauli channel can only be estimated from syndrome data if it is known that the Pauli errors are not correlated across too many qubits, quantified by the \emph{pure} distance.
This limit on correlations can be quite strict, as can be seen for the toric code, which has a pure distance of $d = 4$ independent of system size.
Hence, this assumption is violated by natural noise processes such as error propagation in the stabilizer measurements, which can introduce data errors on all participating qubits.

In this work, we show that the estimation of error rates is possible under much more practical conditions if one focuses only on information that is actually relevant for \ac{QEC}.
It is not necessary to distinguish between logically equivalent errors.
Thus, it suffices to estimate the logical noise channel instead of the physical one.
At least for phenomenological Pauli noise models, we prove that the situation is as good as one could reasonably hope:
as long as the noise affecting a stabilizer code can be corrected by it, one can also estimate the logical noise channel from the corresponding syndrome measurements.

The proof is based on our general framework \cite{Wagner21PauliChannels}, but extended to consider the logical instead of the physical channel.
Similar to randomized benchmarking, we consider the problem in Fourier space \cite{Flammia2019EfficientEstimation}.
This representation corresponds to a description of the logical channel in terms of moments instead of probabilities.
Exploiting a weak assumption of limited correlations, we can further simplify the description by switching from regular moments to a set of canonical moments.
Both the logical channel and the syndrome measurements can be represented by linear equations on a small set of canonical moments.
By considering the ranks of these two linear systems, we then show that
the syndrome measurements determine the logical channel.
Computing the ranks boils down to counting a specific subset of logical operators of the code, which we solve by employing a recent generalization of the
cleaning lemma \cite{kalachev2022_cleaninglemma} of \ac{QEC}.

\section{Stabilizer codes}
Let us quickly recap the most important features of stabilizer codes for our purposes.
A more thorough introduction can e.g.\ be found in the books \cite{NieChu10,lidar_brun_2013}.
A stabilizer code is described by a commuting subgroup $\stabs \subseteq \pauliphases^n$ of the $n$-qubit Pauli group, called \emph{stabilizer group}.
It must fulfill $-I \not \in \stabs$.
The \emph{codespace} is then the simultaneous $+1$ eigenspace of all the stabilizers.
As is usual in the context of \ac{QEC}, we disregard phases and view $\stabs$ as a subgroup of the \emph{effective Pauli group} $\pauli^n \coloneqq \pauliphases^n / \{\pm 1, \pm i\}$. This is an Abelian group, but the relevant commutation relations of $\pauliphases^n$ can be encoded in the \emph{bicharacter} $\scprod{\argdot}{\argdot}$ on $\pauli^n$, given by
\begin{align}
\scprod{a}{e} \coloneqq \begin{cases} +1, & a \text{ and } e \text{ commute in } \pauliphases^n \\
								-1, & a \text{ and } e \text{ anti-commute in } \pauliphases^n	\end{cases} \, .
\end{align}
By definition, all elements of $\stabs$ act trivially on the encoded states.
We can also consider Pauli operators that map the codespace to itself, but do not necessarily act as the identity.
These form the set $\logops \subseteq \pauli^n$ of \emph{logical operators}.
It can be shown that $\logops$ is exactly the set of Pauli operators that commute with all stabilizers.
Formally, we can express this as the \emph{annihilator} $\stabs^\perp$ of $\stabs$ in $\pauli^n$ under the above bicharacter, i.e.
\begin{equation}
\logops \coloneqq \stabs^\perp \coloneqq \{ l \in \pauli^n \, : \, \scprod{s}{l} = +1 \, \forall s \in \stabs \} \, .
\end{equation}
In particular, we have $\stabs \subseteq \logops$ since each stabilizer is itself a logical operator that implements the logical identity.
If a logical operator (other than a stabilizer) occurs as an error, this cannot be detected and the encoded state is corrupted.
The distance $d$ of a code is defined as the minimal weight of an element of $\logops \setminus \stabs$.
This measures the error correction capabilities of the code.
We call a set of qubits $R \subseteq \{1,\dots,n\}$ \emph{correctable} if it only supports trivial logical operators.
This definition is inspired by the discussions in Refs.~\cite{bravyi2009_distanceindimensiond,Burton2020_transversalgateshypergraphcodes}.
In particular, if $|R| < d$, then $R$ is correctable.
This is however generally not an equivalence, and there can be many correctable regions of size much larger than $d$.
For example, any rectangular region of side length at most $d-1$ on the $d \times d$ toric code is correctable, but contains more than $d$ qubits.

We will focus on phenomenological Pauli noise models, and thus do not take into account the details of error propagation inside the measurement circuits.
We can then consider rounds of error correction, and between two rounds a new Pauli error occurs.
These Pauli errors are described by a quantum channel $P$, which is given by a probability distribution
\begin{equation}
P: \pauli^n \mapsto [0,1] \, .
\end{equation}
Later we will also impose some locality assumptions on this channel.

Standard error correction using a stabilizer code proceeds as follows:
in each round, a set of generators $g_1,\cdots,g_m \in \stabs$ is measured.
Ideally, the state lies in the codespace and thus all measurements return $+1$.
However, if an error $e \in \pauli^n$ occurred beforehand, the outcome of the measurement of $g_i$ is $\scprod{g_i}{e} = \pm 1$.
The collection of measurement outcomes of all generators is called the \emph{syndrome} $S(e)$ of an error $e$.
Based on the syndrome, a decoder tries to guess the error that occurred, and applies it as a correction $r$.
Since errors that only differ by stabilizers are logically equivalent, the ideal decoding strategy for a given syndrome $S$ is to return a maximum likelihood estimate of the form
\begin{equation}
r = \argmax_{e \in \pauli^n \, : \, S(e) = S} \sum_{s \in \stabs} P(es) \, .
\end{equation}
Thus, full knowledge of the physical channel $P$ is not necessary for optimal decoding.
Instead, it is sufficient to know the \emph{logical channel} $P_L$, which we define by averaging $P$ over cosets of $\stabs$
\begin{equation}
P_L: \pauli^n \rightarrow [0,1], \quad
P_L(e) = \frac{1}{|\stabs|}\sum_{s \in \stabs} P(es) \, .
\label{eq:logical_channel}
\end{equation}
We note that often the term \emph{logical channel} is defined to be a map acting only on the logical information, conditioned on each syndrome (e.g.\ Refs.~\cite{Beale_2021,Rahn2002,Chamberland2017}).
That is, if the code encodes $k$ qubits there is one distribution on $\mathsf{P}^k$ for each syndrome.
However, this definition depends on the choice of correction for each syndrome since the state needs to be mapped back to the codespace.
Here, we define the logical channel in a decoder-independent way.
In particular, we only consider ``predecoding'' noise, i.e., the noise before any potential decoding operation.
In other words, our definition, Eq.\ \eqref{eq:logical_channel}, can be viewed as a lift of the logical channels for each syndrome from $\mathsf{P}^k$ to $\mathsf{P}^n$, resulting in a distribution $P_L$, which is constant on cosets of $\stabs$.
In particular, $P_L$ contains all the same information as the logical channels in the usual sense.
In standard error correction, it is assumed that the logical channel is known, and the task is to find a good decoding for each syndrome.
Here, however, we will consider a ``reverse'' problem: given (an estimate of) the syndrome statistics, can we (uniquely) obtain the logical channel $P_L$?
Perhaps surprisingly, we will show that this is possible as long as the noise affecting the code is correctable in a certain sense.

\section{Moments}
\label{sec:Moments}
To tackle this estimation problem, we will first switch our description of $P$ via a Fourier transform.
The Fourier transform $\fourier{f}$ of a function $f: \pauli^n \rightarrow \mathbb{R}$ is defined as
\begin{equation}
\fourier{f} : \pauli^n \rightarrow \mathbb{R}, \quad
\fourier{f}(a) = \sum_{e \in \pauli^n} \scprod{a}{e} f(e) \, .
\end{equation}
This is also sometimes called Walsh-Hadmard transform 
\cite{Flammia2019EfficientEstimation}.
From the definition, we see that for any stabilizer $s \in \stabs$, $\fourier{P}(s)$ is exactly the expectation of $s$ in repeated rounds of error correction. It can thus be computed from the measured syndrome statistics.
In analogy, we denote $E = \fourier{P}$ and call this the set of moments, i.e. there is one moment $E(a)$ for each $a \in \pauli$.
One should however keep in mind that only the moments corresponding to stabilizers can be measured without destroying the encoded information.
Since the Fourier transform is an invertible transformation, with inverse given by
\begin{equation}
\fourierinv{f}(e) = \frac{1}{|\pauli^n|}\sum_{a \in \pauli^n} \scprod{a}{e} f(a)  \, ,
\end{equation}
knowing all moments $E$ is equivalent to knowing the complete error distribution $P$.

Since we are only interested in learning the logical channel, only a subset of all moments needs to be estimated. These are exactly the moments corresponding to logical operators.
To see why this is the case, let us first introduce the convolution on $\pauli^n$. For two functions $f,g: \pauli^n \rightarrow \mathbb{R}$, their convolution is defined by
\begin{equation}
(f \conv g)(e) = \sum_{e' \in \pauli} f(e')g(ee') \, .
\end{equation}
As expected, it can be shown that convolutions transform into products under Fourier transform:
\begin{equation}
\fourier{f \conv g} = \fourier{f} \cdot \fourier{g} \, .
\end{equation}
The logical channel $P_L$, defined in Eq.\ \cref{eq:logical_channel}, can be written as the convolution of the physical channel $P$ with the uniform probability distribution over stabilizers $U_\stabs$,
\begin{equation}
P_L = P \conv U_\stabs \, .
\end{equation}
It is well known that $\fourier{U_\stabs} = \Phi_{\stabs^\perp} = \Phi_\logops$, where $\Phi_{\logops}$ is the indicator function of $\logops$  \cite{mao2005_convolutionalfactorgraphs}.
Therefore the logical channel can be characterized in Fourier space by the moments
\begin{equation}
E_\logops \coloneqq E \cdot \Phi_{\logops} \, .
\label{eq:moments-logical}
\end{equation}
This is a special instance of the averaging or subsampling duality explained in \cite{mao2005_convolutionalfactorgraphs}.
To summarize the above discussion, the logical channel is fully characterized by the moments corresponding to logical operators.
The estimation problem can then be phrased as follows:
given the moments $E_\stabs$ of all stabilizers, compute the moments $E_\logops$ of all logical operators.

\section{Correctable noise}
The above estimation problem cannot be solved for arbitrary channels $P$, since in general the moments are independent of each other.
Here, our assumption of limited correlations becomes important.

To formalize this assumption, consider a set of \emph{supports} $\Gamma \subseteq 2^{\{1,\dots,n\}}$, where $2^{\{1,\dots,n\}}$ denotes the powerset of ${\{1,\dots,n\}}$.
These supports are allowed to overlap with each other.
We assume that on each support $\gamma \in \Gamma$, there acts an independent Pauli channel $P_\gamma: \pauli^\gamma \rightarrow [0,1]$.
Thus, the noise is correlated across each support, but not between different supports.
If the supports are small, any high weight error must arise as a combination of independent lower weight errors.
This is the scenario where error correction has a chance to improve the fidelity.
On the other hand, if the supports are too large, error correction usually fails.
Thus, we assume that the noise is \emph{correctable} in the following sense.
\begin{definition}[Correctable noise]
\label{def:CorrectableNoise-MainText}
A Pauli channel $P$ described by a set of supports $\Gamma \subseteq 2^{\{1,\dots,n\}}$ is called \emph{correctable} if the following two conditions are fulfilled:
\begin{enumerate}[label=(\roman*)]
\item For all $\gamma_1,\gamma_2 \in \Gamma$, the union $\gamma_1 \cup \gamma_2$ is a correctable region.
\item \label{item:noise_level} $P_\gamma(I) > \frac{1}{2}$ for all $\gamma \in \Gamma$.
\end{enumerate}
\end{definition}

We see from the definition of distance that the first condition is fulfilled in particular if $|\gamma| \leq \lfloor \frac{d-1}{2} \rfloor$ for all $\gamma \in \Gamma$.
The second condition simply states that the error rates should not be too large.
It guarantees that all moments are positive, i.e. $E(a) > 0$ for all $a \in \pauli^n$.
We would like to emphasize that our definition of correctable noise is quite a weak one: actual \ac{QEC} requires the noise level to be below some code-dependent threshold, which is always  lower than the one imposed by our condition \ref{item:noise_level}.
\Cref{def:CorrectableNoise-MainText} is also distinct from the Knill-Laflamme condition \cite{Knill00TheoryOfQuantum}, \cite[Theorem~10.1]{NieChu10}, which is usually applied to a subnormalized part of the full error channel.

Since the multiplication of independent Pauli random variables corresponds to a convolution of their probability distributions, the full channel $P$ can be written as a convolution of the independent local channels,
\begin{equation}
P = \Conv_{\gamma \in \Gamma} P_{\gamma}.
\end{equation}
In this notation, we set $P_\gamma(e) = 0$ if $\supp(e) \not \subseteq \gamma$.
In order to better capture this structure in Fourier space, we can introduce a set of \emph{canonical moments} $F$ (which we called ``transformed moments'' before \cite{Wagner21PauliChannels}).
For $a,b \in \pauli^n$, let us write $b \leq a$ if $b$ is a substring of $a$.
Then we define the canonical moments as
\begin{equation}
\label{eq:CanonicalMomentsDef-MainText}
F: \pauli^n \rightarrow \mathbb R \, , \quad
F(a) = \prod_{b \in \pauli^n : b \leq a} E(b)^{\mu(b,a)} \, ,
\end{equation}
where $\mu$ is the Möbius function defined by
\begin{equation}
\mu(b,a) = \begin{cases} (-1)^{|a| - |b|}, & b \leq a \\
0, & \text{otherwise} \end{cases} \, ,
\end{equation}
which is well known in combinatorics \cite{Aigner2007_ACourseInEnumeration}.
The Möbius function is defined in such a way that in Eq.\ \cref{eq:CanonicalMomentsDef-MainText}, we divide out that part of the moment $E(a)$ that is already described by substrings $b \leq a$, without ``double counting'' any substring.
Essentially, while the regular moments $E$ also capture correlations across all subsets of their support, the canonical moments only capture correlations across their full support.
The advantage is that a small set of canonical moments is sufficient to fully describe the channel.
In particular, the following two facts about canonical moments are shown in the
\cref{lem:CanonicalMomentProperties} in the appendix.
First of all, we only need to consider the canonical moments that lie completely inside a channel support $\gamma$, since
\begin{equation}
F(a) = 1 \text{ if } \supp(a) \not \subseteq \gamma  \quad \forall \gamma \in \Gamma \, .
\end{equation}
The set of such canonical moments is  $F_{\Gamma'} = [F(a)]_{a \in \Gamma'}$, where
\begin{equation}\label{eq:Gamma-prime}
\Gamma' = \{a \in \pauli \, : \, \exists \gamma \in \Gamma \text{ such that } \supp(a) \subseteq \gamma\} \, .
\end{equation}
Furthermore, the regular moments $E$ are obtained from the canonical moments $F$ by
\begin{equation}
E(a) = \prod_{b \leq a} F(b) \,  .
\end{equation}

\section{Identifiability}
Since the moments $E_\stabs$ can be obtained from the syndrome measurements, and the channel is fully described by the canonical moments $F_{\Gamma'}$, estimation of the physical channel boils down to solving the system of equations
\begin{equation}
\label{eq:EqSystem1-MainText}
E(s) = \prod_{a \in \Gamma', a \subseteq s} F(a) \, .
\end{equation}
For correctable noise, all moments are positive.
Then, Eq.\ \cref{eq:EqSystem1-MainText} can be transformed into a system of linear equations by taking logarithms.
This system can be expressed by the coefficient matrix $D_\stabs$, whose rows are labeled by stabilizers and whose columns are labeled by elements of $\Gamma'$, with entries
\begin{equation}
D_\stabs[s,a] = \begin{cases} 1, & a \subseteq s \\ 0, & \text{ otherwise } \end{cases} \, .
\end{equation}
As we have proven before \cite{Wagner21PauliChannels}, a unique solution exists if the range of correlations of the error channel $P$ is smaller than the pure distance of the code.
Correctable noise generally does not fulfill this strict condition.
Thus, the system is underdetermined and the physical channel $P$ cannot be estimated just from the syndrome measurements.

We are, however, only interested in estimating the logical channel, Eq.\ \eqref{eq:logical_channel}, which contains less information.
As derived above, Eq.\ \cref{sec:Moments},
it suffices to consider the moments $E_\logops$.
The question is now whether the moments $E_\logops$ can be computed from the measured moments $E_\stabs$, i.e., whether the corresponding equations of the form Eq.\ \cref{eq:EqSystem1-MainText} are linearly dependent after taking logarithms.
In other words, the logical channel can be uniquely estimated from the syndrome measurements if
\begin{equation}
\rank(D_\stabs) = \rank(D_\logops) \, .
\end{equation}
This condition is equivalent to $\rank(D_\stabs^TD_\stabs) = \rank(D_\logops^TD_\logops)$.
We will prove this by showing the even stronger statement
\begin{equation}
\label{eq:CoefficientProportional}
D_\stabs^TD_\stabs \propto D_\logops^TD_\logops \, .
\end{equation}

First, note that $D_\stabs^TD_\stabs$ can be easily computed from its definition,
\begin{equation}
D_\stabs^TD_\stabs[a,b] = |\{s \in \stabs \, : \, a \leq s \text{ and } b \leq s \}| \, .
\end{equation}
The analogous statement holds for $D_\logops$.
By rewriting Eq.\ \cref{eq:CoefficientProportional} in terms of individual entries, we see that the logical channel can be uniquely estimated from the syndrome statistics if for all $a,b \in \Gamma'$,
\begin{equation}
\label{eq:CoefficientProportional-explicit}
|\{s \in \stabs \, : \, a,b \leq s\}|
= c\,  |\{l \in \stabs^\perp \, : \, a,b \leq l \}| \, ,
\end{equation}
where $c$ is a constant independent of $a,b$.
This is a counting problem that depends only on global properties of the stabilizers and logical operators, but not on their specific form.
To solve this counting problem, we will employ the well-known cleaning lemma, which was first stated by \citet{bravyi2009_distanceindimensiond}.
Informally, this lemma states that any correctable region can be cleaned from logical operators.
\begin{lemma}[Cleaning Lemma]
Let $R$ be a correctable region. Then any coset $[l] \in \logops / \stabs$ of logical operators  has a representative $l \in \logops$ that has no support on $R$, i.e.\ $\supp(l) \cap R = \emptyset$.
\end{lemma}

Using this lemma, we can prove Eq.\ \cref{eq:CoefficientProportional-explicit}. For all $a,b \in \Gamma'$ we have
\begin{align*}
&|\{l \in \logops \, : \, a \leq l \text{ and } b \leq l \}| \\
&= \sum_{l \in \logops} [a \leq l \text{ and } b \leq l] \\
&= \sum_{[l] \in (\logops / \stabs)} \sum_{s \in \stabs} [a \leq ls \text{ and } b \leq ls] \\
&= \sum_{[l] \in (\logops / \stabs)} \sum_{s \in \stabs} [a \leq s \text{ and } b \leq s] \\
&= |\logops / \stabs| \cdot |\{s \in \stabs \, : \, a \leq s \text{ and } b \leq s \}| \,.
\end{align*}
In the second equality, we split the total sum into smaller sums over logically equivalent subsets of logical operators. Then, the third equality follows from the cleaning lemma:
since $a$ and $b$ correspond to canonical moments, they must be fully contained in some supports $\gamma_a, \gamma_b \in \Gamma$.
For correctable noise, $\gamma_a \cup \gamma_b$ is a correctable region.
Thus, if the union of the supports of $a$ and $b$ is fully contained in $\gamma_a \cup \gamma_b$, it must also be a correctable region.
By the cleaning lemma, we can  choose the representative $l$ of  the coset $[l]$ such that it acts trivially on that region.
Then, $a$ is a substring of $ls$ if and only if it is a substring of $s$, and the same holds for $b$.
This finishes the proof of Eq.\ \cref{eq:CoefficientProportional-explicit}.

We can summarize the discussion of the main text in the following theorem.
\begin{theorem}
\label{Thm:Main-MainText}
A Pauli channel $P$ can be estimated up to logical equivalence from the syndrome measurements of a stabilizer code if $P$ is correctable in the sense of \cref{def:CorrectableNoise-MainText}.
\end{theorem}

Note that while we focused on on stabilizer codes with perfect measurements for simplicity, several generalizations of this result are possible.
Measurement errors can be incorporated using the framework of quantum data syndrome codes \cite{ashikhmin_quantumdatasyndromecodes}.
Furthermore, we can also consider subsystem codes \cite{poulin2005_subsystemcodes}, which generalize stabilizer codes by allowing for some noncommuting measurements.
A full account of these generalizations, including all proofs that are omitted in the main text, is given in the appendix.
The main theorem presented there might also be interesting in contexts other than \ac{QEC}.

\section{Conclusion}
We have shown that the measurements performed during \ac{QEC} contain enough information to estimate a large class of phenomenological Pauli noise models
up to logical equivalence.
Informally, as long as the code can correct the noise, it can also be estimated from the syndrome measurements.
This result opens up new characterization possibilities since the previous results have focused only on estimating physical channels.
Our result applies to data syndrome codes and general subsystem codes, which encompass most codes in the literature.

While the focus of this work is on the fundamental identifiability of Pauli noise from syndrome data in the setting of general subsystem codes, our proofs also suggest a concrete estimation scheme. Since it is sufficient to consider as many equations as there are free parameters in the polynomial system, Eq.\ \eqref{eq:EqSystem1-MainText}, this system can in principle be solved in polynomial time in the code size. The sample complexity, however, depends on the conditioning of this system, and hence on the specifics of the code. We note, however, that for e.g.\ topological codes, estimation is expected to be possible from local subregions of the code, which implies an efficient sampling complexity \cite{Wagner21PauliChannels}.
In order to work out these ideas, a specific analysis of concrete codes is required, which is ongoing research.

We have considered only phenomenological noise models. 
For quantum communication or storage, this might be a
reasonable assumption.
In the context of fault-tolerant quantum computing, however, full circuit level noise models are more realistic than phenomenological ones, which introduces additional complications already for decoding in the first place.
A common approach to this problem is to consider approximate noise models.
For example, a minimum-weight perfect matching decoder maps the actual noise to a simplified graph with weighted edges \cite{obrien_adaptiveweightestimator,hollenberg2013_scqchigherrorrates}.
Here, our results apply directly, and the edge weights can be estimated up to logical equivalence by solving our equation system, Eq.\  \eqref{eq:EqSystem1-MainText}.

The situation is less clear if one is interested in more details than such an effective noise model provides.
In this case, one might attempt to transfer our results using a cutoff for late errors, following \citet{delfosse_beyondsingleshotfaulttolerance}, or using a mapping from circuit noise to subsystem codes, as given in Refs.~\cite{bacon2105_quantumcodesfromcircuits,pryadko2020_maximumlikelihoodcircuitnoise,chubb2019_statisticalmechanicsqec}.
We think that our work can serve as a basis for many possible research questions on characterization in the context of \ac{QEC}.

\begin{acknowledgments}
This work was funded by the Deutsche Forschungsgemeinschaft (DFG, German Research Foundation) under Germany's Excellence Strategy -- Cluster of Excellence Matter and Light for Quantum Computing (ML4Q) EXC 2004/1 -- 390534769.
The work of D. B.\ and H. K.\ is also supported by the German Federal Ministry of Education and Research (BMBF) within the
funding program ``Quantum technologies -- from basic research to market'' in the joint project QSolid (grant
no. 13N16163).
The work of M. K.\ is also supported by the DFG via the Emmy Noether program (grant no. 441423094) and by the German Federal Ministry of Education and Research (BMBF) within the funding program ``Quantum technologies -- from basic research to market'' in the joint project MIQRO (grant no. 13N15522).
\end{acknowledgments}

\appendix

\newpage
\section*{Appendix}
In this appendix, we give a self contained account of our results in a generalized setting.
The core arguments are similar to the main text.
The main difference is that we distinguish between the set of accessible measurements and the set of stabilizers (or gauge group), which describes logical equivalence.
Consequently, a more general version of the cleaning lemma is needed.
All proofs omitted in the main text are also provided in this generalized setting.
Finally, we apply the result to the classes of subsystem codes, which encompasses most quantum-error correction codes that have been constructed, and quantum data-syndrome codes, which allows for a treatment of measurement errors.

\subsection*{Notation}
We denote as $[n] \coloneqq \{1,\dots,n\}$ the set of the first $n$
positive integers.
The field with two elements is denoted $\mathbb{F}_2$.
For a statement $Q$, we denote with $\kron{Q}$ the Iverson bracket of $Q$, which takes the value $1$ if $Q$ is true and $0$ if $Q$ is false.
The powerset of a set $A$ is the set of all subsets of $A$, including the empty set, and it is denoted as $2^A$.
We denote the four Pauli matrices as $I = \bigl( \begin{smallmatrix}
1 & 0 \\ 0 & 1 \end{smallmatrix} \bigr)$, $X = \bigl( \begin{smallmatrix}
0 & 1 \\ 1 & 0 \end{smallmatrix} \bigr)$, $Y = \bigl( \begin{smallmatrix} 0 & -i \\ i & 0 \end{smallmatrix} \bigr)$ and $Z = \bigl( \begin{smallmatrix} 1 & 0 \\ 0 & -1 \end{smallmatrix} \bigr)$.
We also use $I$ for the generic identity matrix, or a generic identity element of a group.

\subsection{Mathematical background}
\label{sec:mathematicalbackground}
For the discussions of stabilizer quantum-error correction, some background on the Pauli group will be useful.
The \emph{$n$-qubit Pauli group}
$\pauliphases^n$ is the group
generated by tensor products of Pauli operators and
the imaginary unit, i.e.,
\begin{equation}
\pauliphases \coloneqq \{\alpha \bigotimes_{i=1}^n e_i \, : \, e_i = \{I,X,Z,Y\}, \alpha \in \{ \pm 1, \pm \i \}  \} \, .
\end{equation}
Since
phases can
often
be disregarded, we also work with
the \emph{effective Pauli group}
\begin{equation}
\pauli^n \coloneqq \pauliphases / \{\pm 1, \pm i\} \, .
\end{equation}
This in an Abelian group.

In \ac{QEC}, errors and stabilizer measurements are often described via an isomorphism $\mathsf P^n \rightarrow \Ftwo^{2n}$ and a scalar product on $\Ftwo^{2n}$.
We will not make use of this identification, and instead express these concepts using group characters of finite Abelian groups.
We give a short introduction here and collect the most important facts for our purposes.
A more thorough description can be found for example in
Refs.~\cite{mao2005_convolutionalfactorgraphs,fulton2013_representation}.

A \emph{group character} of a finite Abelian group $A$ is a group homomorphism
\begin{equation}
\chi: A \rightarrow S^1  \, ,
\end{equation}
where $S^1 \coloneqq \Set*{ c \in \mathbb{C} \given \abs c = 1 }$ is the unit circle.
Group characters themselves form a group $\widehat A$ under pointwise multiplication, called the dual group of $A$.
Pontryagain duality guarantees that for any locally compact and hence for any finite Abelian group, $\widehat A$ is
isomorphic to $A$.
Thus, we can express group characters by elements of the original group.
This notion can be expressed by a \emph{bicharacter}.
In the context of \ac{QEC}, bicharacters
express measurement outcomes of stabilizer measurements.

\begin{definition}
A \emph{bicharacter} of a finite Abelian group $A$ is a map
\begin{equation}
\scprod{}{}: A \times A \rightarrow S^1 \, ,
\end{equation}
such that the map $a \mapsto \scprod{a}{}$ is an isomorphism of $A$ and $\widehat A$.
\end{definition}

This is similar to a scalar product, although we often have $ \scprod{a}{a} = +1$.
Thus, we also have a notion of ``orthogonal complement'', which is the annihilator.
For a subgroup $B \subseteq A$, we define the \emph{annihilator} of $A$ as
\begin{equation}
B^\perp = \{ a \in A \, : \, \scprod{a}{b} = +1 \, \forall b \in B \} \, .
\end{equation}
We always have $(B^\perp)^\perp = B$.
Furthermore, taking the annihilator reverses the order of inclusions.
That is, for any two subgroups $B,C \subseteq A$, if $B \subseteq C$, then $C^\perp \subseteq B^\perp$.
In contrast to a scalar product, it is possible that $B \subseteq B^\perp$.

Using the bicharacter $\scprod{}{}$, we can define the \emph{Fourier transform} of a map $f: A \rightarrow \mathbb{C}$ as
\begin{equation}
\label{eq:FourierDef}
\begin{split}
&\fourier{f} : A \rightarrow \mathbb{C} \\
&\fourier{f}(a) = \sum_{b \in A} \scprod{a}{b} f(b) \, .
\end{split}
\end{equation}
This is an invertible transformation with inverse
\begin{equation}
\begin{split}
&\fourierinv{f} : A \rightarrow \mathbb{C} \\
&\fourierinv{f}(a) = \frac{1}{|A|}\sum_{b \in A} \scprod{b}{a^{-1}} f(b) \, .
\end{split}
\end{equation}
Furthermore, we will use the convolution of two maps $f,g: A \rightarrow \mathbb{C}$, which is defined as
\begin{equation}
(f \conv g)(a) = \sum_{b \in A} f(b)g(ab^{-1}) \, .
\end{equation}
As expected, convolutions
are mapped
to products
by the
Fourier transform, i.e.
\begin{equation}
\fourier{f \conv g} = \fourier{f} \cdot \fourier{g} \, .
\end{equation}

For any subgroup $B \subseteq A$, we denote with $\Phi_B$ the indicator function of $B$, i.e. $\Phi_B(a) = 1$ if $a \in B$ and $\Phi_B(a) = 0$ otherwise.
Furthermore, we denote the scaled indicator function as $U_B \coloneqq \frac{1}{|B|}\Phi_B$,
which is the uniform probability distribution on $B$.
It can be shown that the following duality holds.

\begin{lemma}[Ref.~\cite{mao2005_convolutionalfactorgraphs}]
\label{lem:IndicatorDuality}
For any subgroup $B \subseteq A$ of an Abelian group $A$:
\begin{align}
&\fourier{U_{B}} = \Phi_{B^\perp} \\
&\fourier{\Phi_B} = \abs{A}U_{B^\perp} \, .
\end{align}
\end{lemma}

All important groups considered in this work have a direct product structure, i.e.
\begin{equation}
      A = \prod_{i=1}^n A_i\, .
\end{equation}
We will then always use the \emph{product bicharacter} on $A$, which is given by the product of bicharacters on the $A_i$,
\begin{equation}
\scprod{a}{b} = \prod_{i=1}^n \scprod{a_i}{b_i}
\end{equation}
for any $a = (a_1,a_2,\dots, a_n)\in A$ and similar $b$.
The \emph{support} of an element $a \in A$ is
\begin{equation}
\supp(a) = \{ i \in [n] \, : \, a_i \neq I  \} \, .
\end{equation}
We will say that $a$ is \emph{supported} on a region $R \subseteq [n]$ if $\supp(a) \subseteq R$.
The corresponding subgroup to a region $R$ is denoted as $A_R \coloneqq \prod_{i \in R} A_i$.
This is naturally embedded as a subgroup in $A$.
The \emph{complement} of $R \subseteq [n]$ is denoted as $R^c = [n] \setminus R$.
If we use the product bicharacter on $A$, we have that
\begin{equation}
\label{eq:ComplementVsAnnihilator}
A_R^\perp = A_{R^c} \, .
\end{equation}
Given an element $a \in A$, we denote with $a_R$ its restriction to $R$, i.e.\ $a_R = a$ on $R$ and $a_R = 1$ on $R^c$.

Finally, we will be interested in functions with local support.
Given a function $f_R: A_R \rightarrow \mathbb{C}$, there are two important ways to extend it to a function $f: A \rightarrow \mathbb{C}$.
The first is to set $f(a) = 0$ if $a \not \in A_R$.
This is called the \emph{impulsive} extension.
The second is the to set $f(a) = f_R(a_R)$, which is called \emph{periodic} extension.
These two possibilities transform into each other under Fourier transform.

\begin{lemma}[Ref.~\cite{mao2005_convolutionalfactorgraphs}]
\label{lem:ImpulsivePeriodicDuality}
Let $f_R: A_R \rightarrow \mathbb{C}$ and $g_R: A_R \rightarrow \mathbb{C}$ be its Fourier transform (on $A_R$).
Let $f$ be the impulsive extension of $f_R$ and $g$ be the periodic extension of $g_R$.
Then $\fourier{f} = g$.
\end{lemma}

We will mainly work with three groups.
These are the effective Pauli group $\pauli^n$, the group of bit-strings $\Ftwo^m$, and their direct product
$\paulids \coloneqq \pauli^n \times \Ftwo^m$.
Since all elements of these groups have order two, bicharacters of these groups will only take values $\pm 1$.

For $\pauli^n$, the bicharacter encodes commutation relations, and is also called scalar commutator,
\begin{align}
\scprod{a}{e} = \begin{cases} +1, & a \text{ and } e \text{ commmute in } \pauliphases^n \\
                -1, & a \text{ and } e \text{ anti-commmute in } \pauliphases^n \end{cases} \, .
\end{align}
Note that this is the product bicharacter when we view $\mathsf{P}^n = \prod_{i=1}^n \mathsf{P}^1$.
On $\Ftwo^m$, we use the bicharacter that is related to the usual scalar product,
\begin{equation}
\scprod{e}{f} = (-1)^{\sum_i e_if_i} \, ,
\end{equation}
and again this coincides with the product bicharacter.
Finally, on $\paulids$, we directly use the product bicharacter
\begin{equation}
\label{eq:BicharacterPauliDS}
\scprod{(a,e)}{(b,f)} = \scprod{a}{b} \cdot \scprod{e}{f} \, .
\end{equation}

\subsection{Setting and main result}
Now, we state and prove our main result in an abstract setting first.
For ease of exposition, we still stick to terminology close to that of \ac{QEC}.

We consider the group $A = \prod_{i = 1}^n A_i$, where each $A_i$ is either $\pauli$ or $\Ftwo$.
It comes equipped with the product bicharacter.
Both errors and measurements are described as element of $A$.

We are interested in estimating an error channel described by a probability distribution $P: A \rightarrow [0,1]$.
For this purpose,
we assume that we have access to a group of \emph{measurements} $\meas \subseteq A$.
The assumption that the measurements from a group is relatively weak.
In the context of \ac{QEC}, we measure a set of generators and all other outcomes are defined by products of the generator outcomes.
We will perform multiple rounds of measurements, and assume that before each round an independent error $e \in A$ occurs.
The outcome of measurement $s \in \meas$ is described by $\scprod{s}{e}$.
We will refer to this as a \emph{phenomenological} noise model, since errors are independent and identically distributed between rounds and no new errors arise during the round of measurements.
Errors that give a $+1$ outcome for every measurement $s \in \meas$ are called \emph{undetectable}.
The set of undetectable errors is exactly $\undet = \meas^\perp$.
Furthermore, we will only be interested in estimating the channel up to some logical equivalence, described by a subgroup $\gauge \subseteq A$ which we will call \emph{gauge group}.
An overview of these groups and their relations is give in \cref{fig:groups}.
Errors differing only by an element $s \in \gauge$ are considered logically equivalent.
More precisely, we are interested in estimating the \emph{logical channel} $P_L$ which is obtained by averaging over cosets of $\gauge$, resulting in
\begin{equation}
\begin{split}
&P_L: A \rightarrow [0,1] \, ,\\
&P_L(e) = \frac{1}{|\gauge|} \sum_{s \in \gauge} P(es) \, .
\end{split}
\end{equation}
In the setting of stabilizer codes, $P_L$ describes the action of the noise on the encoded information.
The logical channel can be conveniently expressed as a the convolution
\begin{equation}
\label{eq:LogChannelConvolution}
P_L = P \conv U_{\gauge}
\end{equation}
of $P$ with the uniform distribution $U_\gauge$.
Complementary to the gauge group, we define $\logops \coloneqq \gauge^\perp$ and call this the set of \emph{logical operators}.
Finally, we will require that these sets are related by the dual inclusion relations $\gauge \subseteq \undet$ and $\meas \subseteq \logops$ (one implies the other by taking annihilators). The condition $\gauge \subseteq \undet$ means that logically equivalent errors must have the same measurement outcomes.

\begin{figure}
\tikzset{
  symbol/.style={
    draw=none,
    every to/.append style={
      edge node={node [sloped, allow upside down, auto=false]{$#1$}}}
  }
}
\begin{tikzcd}
&\gauge \arrow[d, leftrightarrow, "\perp"'] \arrow[r, symbol=\subseteq] &\undet \\
&\logops \arrow[r, symbol=\supseteq]  &\meas \arrow[u, leftrightarrow, "\perp"']
\end{tikzcd}
\caption{An overview over the abstract setting, described by four groups.
The main ingredients are a group of measurements $\meas$ and a gauge group $\gauge$.
The gauge group describes which errors are considered logically trivial.
The annihilator of $\meas$ is the group of undetectable errors $\undet=\meas^\perp$, and the annihilator of $\gauge$ is the group of logical operators $\logops=\gauge^\perp$.
The groups fulfill the dual inclusions $\gauge \subseteq \undet$ and $\meas \subseteq \logops$.}
\label{fig:groups}
\end{figure}

In this setting, our main result can be stated as follows:
\begin{theorem}
\label{thm:Main}
Let $\gauge,\meas \subseteq A$ be a gauge group and measurement group fulfilling $\gauge \subseteq \meas^\perp$.
If the error channel $P$ is correctable in the sense of \cref{def:CorrectableNoise}, then the logical channel $P_L$ defined by the gauge group $\gauge$ can be uniquely estimated
from the expectations of the measurements $\meas$.
\end{theorem}

To recover the result for stabilizer codes, as treated in the main text, we set $A = \pauli^n$. The set of measurements $\meas$ and the gauge group $\gauge$ are then identical, both equal to the stabilizer  group of the code.
Consequently, the undetectable errors $\undet$ and the logical operators $\logops$ also coincide, and are both given by the logical operators of the code.
Later, we will also explain how to specialize our setting to the more general classes of subsystem codes (\cref{sec:subsystem}) and quantum data-syndrome codes (\cref{sec:data-syndrome}) with phenomenological noise.

\subsection{Moments}
We start the proof of \cref{thm:Main} by describing the estimation problem in Fourier space. We define the moments
\begin{equation}
E = \fourier{P} \, .
\end{equation}
Since the Fourier transform is invertible, the set of all moments $(E(a))_{a \in A}$ fully characterizes the channel $P$.
Furthermore, for an element $s \in \meas$, $E(s)$ is the expectation of the measurement of $s$ in repeated rounds.
Thus, the moments corresponding to $\meas$ can be obtained from our measurements.

The logical channel can also be fully characterized by a subset of moments.
Applying a Fourier transform to \cref{eq:LogChannelConvolution} and using \cref{lem:IndicatorDuality} results in
\begin{equation}
\fourier{P_L} = E \cdot \Phi_{\logops} \, .
\end{equation}
Thus, to obtain the logical channel, we have to compute all the moments corresponding to $\logops$, while we can only measure moments corresponding to $\meas$.
We will see that this is indeed possible, assuming that the error channel is not correlated over too large regions.
These assumptions on the noise are formalized in the next section.

\subsection{Local noise}
To formalize the assumption of limited correlations, we assume the total error in each round is a product of many local errors that occur independently.
The noise then factorizes into a set of local channels, characterized by the corresponding set of local supports $\Gamma \subseteq 2^{[n]}$.
For each $\gamma \in \Gamma$, there is a local channel $P_\gamma: A_\gamma \mapsto [0,1]$.
We extend the local channels impulsively to $A$, i.e. we set $P_\gamma(e) = 0$ if $e \not \in A_\gamma$.
The total error distribution is then given by
\begin{equation}
\label{eq:LocalNoise}
P = \Conv_{\gamma \in \Gamma} P_\gamma \, .
\end{equation}
Denoting $E_\gamma = \fourier{P_\gamma}$, we obtain
\begin{equation}
E = \prod_{\gamma \in \Gamma} E_\gamma \, .
\label{eq:Moments-Decomposition}
\end{equation}
Here, each $E_\gamma$ must be extended periodically from $A_\gamma$ to $A$, due to \cref{lem:ImpulsivePeriodicDuality}.
More explicitly, we have $E_\gamma(e) = E_\gamma(e_\gamma)$.

We assume that the individual regions are small enough to not support logically non-trivial undetectable errors.
To formalize this, we define a notion of correctable region, which is inspired by the setting of topological codes \cite{bravyi2009_distanceindimensiond}.
\begin{definition}
\label{def:CorrectableRegion}
A region $R \subseteq [n]$ is called \emph{correctable} if every undetectable error $e\in\undet$ supported on $R$ is logically trivial, i.e.\ $e\in\gauge$.
\end{definition}

Using this definition, we can state our assumptions on the noise.
\begin{definition}
\label{def:CorrectableNoise}
A channel of the form \eqref{eq:LocalNoise}
is called \emph{correctable} if the following two conditions are fulfilled:
\begin{enumerate}
\item For all $\gamma_1,\gamma_2 \in \Gamma$, $\gamma_1 \cup \gamma_2$ is a correctable region.
\item All moments are positive, i.e.\ $E(a) > 0$ for all $a \in A$.
\end{enumerate}
\end{definition}

The first condition states that correlations can not be so large that uncorrectable errors occur frequently.
Instead uncorrectable errors should only be allowed to occur as a combination of many smaller independent errors.
We will later relate this condition to the distance of a code.
The second condition essentially states that the error rates are not too large.
A sufficient conditions for this is $P_\gamma(I) > \frac{1}{2}$ for all $\gamma \in \Gamma$. An alternative, but less practical, sufficient condition is $P(I) > \frac{1}{2}$.

\subsection{Canonical moments}
The factorization \eqref{eq:Moments-Decomposition} can be used to find a more compact characterization of the moments $E$.
Intuitively, we note that the moment $E(a)$ captures correlations across all substrings of $a$.
In particular, $E(a)$ can be non-trivial even if $a$ is not contained in a support $\gamma \in \Gamma$ of our noise model.
We will find an alternative set of moments $F: A \rightarrow \mathbb{R}$, called \emph{canonical moments}, that only capture correlations across their whole support. In particular, the canonical moments fulfill $F(a) = 1$ if $a$ is not contained in a support.
Thus, a small set of low weight canonical moments is sufficient to fully characterize the channel.

Formally, we define the canonical moments by a Möbius inversion.
Möbius inversion is a generalization of the inclusion-exclusion principle of combinatorics \cite{Aigner2007_ACourseInEnumeration}.
Essentially, we divide out correlations on substrings of $a$ from the moment $E(a)$, while being careful not to double count any substrings.
This leaves only correlations across the full support.
In order to do this, we
consider $A$ as a partially ordered set (poset), where the ordering is the substring relation.

\begin{definition}[substring ordering]
We say $a \in A$ is a \emph{substring} of $b \in A$ if for all $i \in [n]$ either $a_i = I$ or $a_i = b_i$.
In this case we write $a \leq b$.
\end{definition}

We will need the \emph{Möbius function} of this poset.
For our purposes, the Möbius function is defined to be the function fulfilling the following inversion theorem, which can be found e.g.\ in Ref.~\cite[Theorem 5.5]{Aigner2007_ACourseInEnumeration} or \cite{roman2006_fieldtheory} in case of the multiplicative version.

\begin{definition}[Möbius function and Möbius inversion]
\label{def:Moebius}
Let $S$ be a partially ordered set.
The Möbius function $\mu$ of $S$ is the function $\mu: S \times S \mapsto \mathbb{R}$ such that for any two functions $f,g: S \mapsto \mathbb{R}$,
\begin{equation}
f(t) = \prod_{s \leq t} g(s) \, ,
\end{equation}
if and only if
\begin{equation}
g(t) = \prod_{s \leq t} f(a)^{\mu(s,t)} \, .
\end{equation}
\end{definition}

In our setting, we obtain the following.

\begin{lemma}
The Möbius function of $(A, \leq)$ is given by
\begin{equation}
\begin{split}
&\mu: A \times A \rightarrow \mathbb{R}\, , \\
&\mu(b,a) =
\begin{cases}
      (-1)^{|a|-|b|} & \text{if } b \leq a \, ,
      \\
      0 & \text{otherwise. }
\end{cases}
\end{split}
\label{eq:MobiusFunction-def}
\end{equation}
\end{lemma}
\begin{proof}
For any given $a \in A$, the poset $\{b \in A \, : \, b \leq a \}$ is isomorphic to the poset $\{s \subseteq \supp(a) \}$ ordered by set inclusion. The Möbius function of this is well known to be $\mu(s,t) =  (-1)^{|t| - |s|}$.
Alternatively, one can use \cite[Proposition 5.4]{Aigner2007_ACourseInEnumeration} and induction.
\end{proof}
Now, we define the \emph{canonical moments} as
\begin{equation}
F(a) \coloneqq \prod_{b \leq a} E(b)^{\mu(b,a)} \, .
\label{eq:CanonicalMoments-def}
\end{equation}
This definition essentially corresponds to the canonical factorization of a factor graph describing the moments, compare Refs.~\cite{koller_pgm, abbeel_learningfactorgraphs}.
The canonical moments have two important properties.

\begin{lemma}[Properties of canonical moments]
\label{lem:CanonicalMomentProperties}
\hfill
\begin{enumerate}
\item The moments can be expressed by the canonical moments as
\begin{equation}
E(a) = \prod_{b \leq a} F(b) \, .
\end{equation}
\item For any $a \in A$ such that $a$ is not contained in any support, i.e.\ $\supp(a) \not \subseteq \gamma$ for all $\gamma \in \Gamma$, we have $F(a) = 1$.
\end{enumerate}
\end{lemma}

\begin{proof}
The first statement is given by the \cref{def:Moebius} of the Möbius function.

Now, we prove the second statement.
First, from the definition \eqref{eq:CanonicalMoments-def} of the canonical moments $F$ and the decomposition \eqref{eq:Moments-Decomposition} we
obtain
\begin{equation}
F(a) = \prod_{b \leq a} \prod_{\gamma \in \Gamma} E_\gamma(b)^{\mu(b,a)} = \prod_{\gamma \in \Gamma} \prod_{b \leq a} E_\gamma(b)^{\mu(b,a)} \, .
\end{equation}
We can evaluate the second product by splitting
it into products over
$B = \{ b \in A_{\gamma} \, : \, b \leq a \}$ and $B^c = \{c \in A_{\gamma^c} \, : \, c\leq a \}$
as
\begin{align*}
&\prod_{b \leq a} E_\gamma(b)^{\mu(b,a)} = \prod_{b' \in B} \prod_{c' \in B^c} E_\gamma(b')^{\mu((b',c'),a)} \\
&= \prod_{b' \in B} E_\gamma(b')^{\sum_{c' \in B^c} \mu((b',c'),a)} \, ,
\end{align*}
where we have used
the periodicity $E_\gamma(b) = E_\gamma(b_\gamma)$ and have denoted $b = (b',c')\in A = A_\gamma \times A_{\gamma^c}$.
From the explicit expression \eqref{eq:MobiusFunction-def} for the Möbius function we
obtain
\begin{align*}
&\sum_{c' \in B^c} \mu((b',c'),a)
= \mu(b',a) \sum_{c' \in B^c} (-1)^{|c'|} \\
&= \mu(b',a) \sum_{0 \leq w \leq |\gamma^c \cap \supp(a)| } (-1)^{|c|} \binom{|\gamma^c \cap \supp(a)|}{w} \\
&= \mu(b',a) \kron{|\gamma^c \cap \supp(a)| = 0} = \mu(b',a)\kron{\supp(a) \subseteq \gamma}
\,,
\end{align*}
where in the second equality we sorted the elements of $B^c$ by their weight, and the third equality follows from the fact that $\sum_{i=0}^n (-1)^i\binom{n}{i} = \kron{n = 0}$.

Putting everything together proves \cref{lem:CanonicalMomentProperties}.
\end{proof}

We conclude that $E$ and hence the physical channel $P$ is characterized by the set of low weight canonical moments corresponding to
\begin{equation}
\Gamma' \coloneqq \{a \in A \, : \, \exists \gamma \in \Gamma \text{ such that } \supp(a) \subseteq \gamma \} \, .
\label{eq:Gamma-prime_Apdx}
\end{equation}
Indeed,
we can express the regular moments by the equation system
\begin{equation}
\label{eq:EqSystem1}
E(a) = \prod_{b \in \Gamma' : b \leq a} F(b) \, .
\end{equation}
Since we can measure $E(s)$ for $s \in \meas$, the estimation problem can now be phrased in terms of these equations.

\subsection{Rank of the coefficient matrix}
Since we assumed that all moments are positive, \cref{eq:EqSystem1} corresponds to a linear system after taking logarithms.
This linear system can be compactly expressed by a coefficient Matrix $D$, whose rows are labeled regular moments and whose columns are labeled by canonical moments. For $a \in A$ and $b \in \Gamma'$, the corresponding entry of $D$ is
\begin{equation}
D[a,b] = \begin{cases} 1 & b \leq a \\ 0 & \text{ otherwise } \end{cases} \, .
\label{eq:CoefficientMatrix-def}
\end{equation}
In particular, the measurements we can perform are characterized by the submatrix $D_\meas$, whose rows are labeled by element of $\meas$, and the logical channel is similarly characterized by the submatrix $D_{\logops}$.
Note that since $\meas \subseteq \logops$, $D_\meas$ is a submatrix of $D_\logops$.
The estimation problem can be solved if the rows of $D_{\logops}$ are linearly dependent on the rows of $D_\meas$.

Based on these considerations, to proof \cref{thm:Main}, we need to show
\begin{equation}
\rank(D_\meas) = \rank(D_{\logops}) \,
\end{equation}
or equivalently
\begin{equation}
\rank(D_\meas^TD_\meas) = \rank(D_{\logops}^TD_{\logops}) \, .
\label{eq:RankCondition}
\end{equation}

To show \cref{eq:RankCondition}, we first reduce it to a counting problem.
From the definition \eqref{eq:CoefficientMatrix-def} of $D$, we obtain that
\begin{equation}
D_\meas^TD_\meas[a,b] = |\{ s \in \meas \, : \, a \leq s \text{ and } b \leq s \}| \, ,
\end{equation}
and analogously for $D_{\logops}$.
The advantage of this formulation is that we only need to consider global properties of groups $\meas$ and $\logops$, but not the specific form of their elements.
We will prove that $D_\meas^TD_\meas \propto D_{\logops}^TD_{\logops}$, i.e. that
\begin{equation}
\label{eq:Counting}
\begin{split}
|\{ s \in \meas \, : \, a \leq s \text{ and } b \leq s \}| \\
= \alpha |\{ l \in \logops \, : \, a \leq l \text{ and } b \leq l \}| \, ,
\end{split}
\end{equation}
for a constant $\alpha$ independent of $a,b$.

\subsection{A general cleaning lemma}
Our proof of \cref{eq:Counting} relies on an abstract variant of the cleaning lemma which was recently derived by \citet{kalachev2022_cleaninglemma}.
The original cleaning lemma was formulated for stabilizer codes in \cite{bravyi2009_distanceindimensiond}.
It states that any correctable region can be ``cleaned'' of logical operators, i.e.\ any logical operator has a representative that it is supported outside of this correctable region.
Since all logical operators arise from their representatives by multiplication with stabilizers, this means that on each correctable region, the group of logical operators and the stabilizer group look essentially identical.
This then solves the counting problem \eqref{eq:Counting}.

We will derive a similar result in our more general setting, however the specifics are a bit different.
For standard stabilizer codes, the gauge group $\gauge$ and the group of measurements are both given by the stabilizer group $\stabs$. Then $\stabs \subseteq \stabs^\perp = \logops$, and the cleaning lemma is then a statement about the quotient group $\logops / \stabs$.
Since in general, the gauge group and the measurements do not coincide, the situation is more complicated.
Instead of the inclusion $\stabs \subseteq  \logops$, we will consider the dual inclusions $\gauge \subseteq \undet$ and $\meas \subseteq \logops$.
The cleaning lemma will then be a statement about $\logops / \meas$.
One can also find a similar statement about $\undet / \gauge$.

We start with the abstract version of the cleaning lemma mentioned above.
Translated to our setting, it states the following.

\begin{lemma}[abstract cleaning lemma \cite{kalachev2022_cleaninglemma}]
\label{lem:CleaningAbstract}
For any three subgroups $\eta,\xi$ and $\alpha$ of an Abelian group $A$ such that $\xi \subseteq \eta^\perp$, we have
\begin{equation}
|(\eta^\perp \cap \alpha) / (\xi \cap \alpha)| \cdot |(\xi^\perp \cap \alpha^\perp) / (\eta \cap \alpha^\perp)| = |\xi^\perp / \eta| \, .
\end{equation}
\end{lemma}
\begin{proof}
This follows from Theorem 3.10 of \cite{kalachev2022_cleaninglemma}, using the following translation.
The lattice $L$ is the lattice of subgroups of $A$, where the join of two subgroups is the subgroup generated by their union, and the meet of two subgroups is their intersection.
The grading is the $Q^+$-grading given by the size of a group.
The quasi-complementation $\dagger$ is the annihilator $\perp$.
This is similar to the setting of \cite{kalachev2022_cleaninglemma}[Section 5.3].
\end{proof}

As a corollary, we obtain the following ``concrete'' cleaning lemma.

\begin{lemma}[cleaning lemma]
\label{lem:CleaningConcrete}
Let $R \subseteq [n]$ be a correctable region.
Then every coset $[l] \in \logops / \meas$ has a representative $l$ that has no support on $R$, i.e.\  $\supp(l) \cap R = \emptyset$.
\end{lemma}

In order to prove this statement, we make use of a simple technical lemma.

\begin{lemma}
\label{lem:QuotientEmbedding}
Let $\eta$ and $\alpha$ be subgroups of an Abelian group $A$ and $\xi \subseteq \eta$.
Then there is a canonical embedding
\begin{equation}
(\eta \cap \alpha) / (\xi \cap \alpha) \rightarrow \eta / \xi
\end{equation}
\end{lemma}
\begin{proof}
The embedding is defined by mapping the equivalence class $[a] \in (\eta \cap \alpha) / (\xi \cap \alpha)$ to $[a] \in \eta / \xi$.
This map is well-defined:
If $[a] = [b]$ in $(\eta \cap \alpha) / (\xi \cap \alpha)$, then $a = bc$ with $c \in (\xi \cap \alpha) \subseteq \xi$, and thus $[a] = [b]$ in $\eta / \xi$.
Now we show injectivity. If $[a] = 1$ in $\eta / \xi$, then $a \in \xi$.
Since we also have $a \in \alpha$ by definition, it follows $a \in \xi \cap \alpha$, and thus $[a] = 1$ in $(\eta \cap \alpha) / (\xi \cap \alpha)$.
\end{proof}

\begin{proof}[Proof of \cref{lem:CleaningConcrete}]
Since $\gauge \subseteq \undet = \meas^\perp$, we can apply the abstract cleaning \cref{lem:CleaningAbstract} with $\eta = \meas$, $\xi = \gauge$ and $\alpha = A_R$.
We obtain
\begin{equation*}
\begin{split}
&|(\undet \cap A_R) / (\gauge \cap A_R)| \cdot |(\logops \cap (A_R)^\perp) / (\meas \cap (A_R)^\perp)| \\
&= |\logops / \meas| \, .
\end{split}
\end{equation*}
Since $R$ is correctable, the first term is $1$.
Thus,
\begin{equation}
|(\logops \cap (A_R)^\perp) / (\meas \cap (A_R)^\perp)| = |\logops / \meas| \, .
\end{equation}
By \cref{lem:QuotientEmbedding}, the group on the left-hand side is embedded in the group on the right hand side.
Thus this equation implies that they are actually equal.
Since we use the product bicharacter on $A$, $(A_R)^\perp = A_{R^c}$.
Thus any element $[l] \in \logops / \meas$ has a representative $ l \in \logops \cap (A_{R^c})$, i.e. a representative that has no support on $R$.
\end{proof}

\subsection{Cleaning up}
Using the cleaning \cref{lem:CleaningConcrete}, we can now finish the proof of  \cref{eq:Counting}, and thus of \cref{thm:Main},
using similar arguments to the stabilizer code case.
\begin{proof}[Proof of \cref{thm:Main}]
Let $a,b \in \Gamma'$, i.e. $a$ and $b$ correspond to non-trivial canonical moments. Then we have,
\begin{equation}
\begin{aligned}
&|\{l \in \logops \, : \, a \leq l \text{ and } b \leq l \}| \\
&= \sum_{l \in \logops} [a \leq l \text{ and } b \leq l] \\
&= \sum_{[l] \in (\logops / \meas)} \sum_{s \in \meas} [a \leq ls \text{ and } b \leq ls] \\
&= \sum_{[l] \in (\logops / \meas)} \sum_{s \in \meas} [a \leq s \text{ and } b \leq s] \\
&= |\logops / \meas| \cdot |\{s \in \meas \, : \, a \leq s \text{ and } b \leq s \}| \,,
\end{aligned}
\end{equation}
where we have used the following steps.
In the second equality, we split the sum into a sum over cosets $[l]$ of $\meas$, where each coset is described by a representative $l \in \logops$.
The third equality used the cleaning \cref{lem:CleaningConcrete} in the following way:
By the properties of $\Gamma'$ from \cref{eq:Gamma-prime_Apdx},
the support of $a$ and $b$ must be contained in supports $\gamma_a, \gamma_b \in \Gamma$.
Then, by the assumption that the noise is correctable (\cref{def:CorrectableNoise}), $\supp(a) \cup \supp(b)$ must be a correctable region.
Thus, by \cref{lem:CleaningConcrete}, we can always choose the representative $l$ such that it has no support on $\supp(a) \cup \supp(b)$. Then the substring relations $a \leq ls$ and $b \leq ls$ are only determined by $s$. This finishes the proof of \cref{thm:Main}.
\end{proof}

We will now discuss some specializations of this theorem for different classes of \ac{QEC} codes.
The case of stabilizer codes was already treated in detail in the main text.
We can in fact treat even more general classes of codes.

\subsection{Subsystem codes}
\label{sec:subsystem}
Subsystem codes \cite{lidar_brun_2013, poulin2005_subsystemcodes} are an important generalization of stabilizer codes.
We will explain the basic principles following \cite{vuillot2019_latticesurgeryisgaugefixing}.
A subsystem code can be viewed as a stabilizer code where some of the logical qubits are not used to encode information.
The corresponding logical operators can be measured without destroying the encoded information.
The primary advantage is that this can often lead to stabilizer measurements of lower weights.
Furthermore, some fault-tolerant schemes are naturally described in the language of subsystem codes \cite{vuillot2019_latticesurgeryisgaugefixing}.
Finally, the effect of circuit noise can also be expressed in the language of subsystem codes \cite{pryadko2020_maximumlikelihoodcircuitnoise, chubb2019_statisticalmechanicsqec, bacon2105_quantumcodesfromcircuits}.
Thus, there is some hope that the following results can also be used to treat circuit noise models instead of phenomenological noise models for stabilizer codes.

A subsystem code can be described by a gauge group $\gauge \subseteq \pauli^n$, whose elements act trivially on the encoded information.
The gauge group contains the stabilizers as well as the logical operators that only act on the unused logical qubits.
Unless the code is a standard stabilizer code, $\gauge$ is not Abelian (when viewed as a subgroup of $\pauliphases^n$).
In the effective Pauli group, this can be expressed as $\gauge \not \subseteq \gauge^\perp$.
The stabilizer group is then the center of the gauge group in $\pauliphases^n$. Expressed in the effective Pauli group:
\begin{equation}
\label{eq:stabs-subsystem}
\stabs \coloneqq \gauge^\perp \cap \mathscr G \, .
\end{equation}
As for stabilizer codes, error detection is performed by measuring the stabilizer group in each round, resulting in a set of $\pm 1$ outcomes called the syndrome.
This can be done by either measuring a set of generators of $\stabs$, or by splitting the stabilizers into products of possibly non-commuting gauge operators and measuring these gauge operators.
Since the gauge operators do not affect the encoded information, the fact that these measurements do not commute does not affect the encoded information,
and it might allow for measurements with lower weight than the stabilizer generators.

Operators that affect the encoded information without being detected are called \emph{(dressed) logical operators}.
The set of such operators is given by $\logops_d = \stabs^\perp$.
If the operators act only on the actual logical qubits, but not on the discarded gauge logical qubits, then they are called \emph{bare logical operators}. The group of bare logical operators is $\logops_b = \gauge^\perp$.
As usual, the distance of the code is defined as the smallest weight of an undetectable error that non-trivially affects the logical information, i.e.\ as the minimal weight of an element of $\logops_d \setminus \gauge$.

We consider a subsystem code subject to phenomenological noise, where before each error correction round an error $e \in \pauli^n$ occurs according to some distribution $P(e)$.
Here, we assume perfect measurements.
In each round, the syndrome of the data error that was accumulated over all previous rounds is measured, and thus the errors in different rounds are not independent.
However, if we consider the syndromes relative to the syndrome of the previous round, then we only detect the new errors.
The same effect is achieved by tracking the Pauli frame, or by applying a correction between rounds that returns the state to the code space.
Thus, in each round we can obtain the measurement outcomes $\scprod{s}{e}$ for each $s \in \stabs$.
This means that set of available measurements $\meas$ is exactly the stabilizer group $\stabs$.

To summarize, we can apply \cref{thm:Main} to subsystem codes by setting $\gauge = \gauge$, $\meas = \stabs$, $\undet = \logops_d$ and $\logops = \logops_b$.
Note that \cref{eq:stabs-subsystem} implies the inclusion $\meas \subseteq \gauge^\perp$, as required by \cref{thm:Main}.
Furthermore, the definition of distance implies that any region of size at most $d-1$ is correctable in the sense of \cref{def:CorrectableRegion}.
Thus, if the error distribution factorizes into independent channels $P_\gamma$, such that each support $\gamma$ contains no more than $\lfloor \frac{d-1}{2} \rfloor$ qubits, the first part of \cref{def:CorrectableNoise} is also fulfilled.
If furthermore $P_\gamma(I) > \frac{1}{2}$ for all $\gamma$, then $E_\gamma > 0$ for all $\gamma$ and thus $E = \prod_\gamma E_\gamma > 0$.
Then, the channel $P$ is correctable in the sense of \cref{def:CorrectableNoise}.
We obtain the following corollary.

\begin{corollary}
Phenomenological data noise with error rates smaller than $\frac{1}{2}$ can be estimated up to logical equivalence from the measurements of a subsystem code if the noise is not correlated over more than half the distance of the code.
\end{corollary}

\subsection{Quantum data-syndrome codes}
\label{sec:data-syndrome}
So far, we have only treated data errors and assumed perfect measurements.
Now, we will consider measurement errors in a phenomenological noise model.
A simple framework for this is provided by quantum data-syndrome codes \cite{ashikhmin_quantumdatasyndromecodes, fujiwara_datasyndromecodes}, which allow for a unified treatment of data and measurement errors.
It should however be noted that, while quantum data-syndrome codes capture a large class of fault-tolerant measurement schemes, some adaptive schemes such as flag fault-tolerance  \cite{chao2018_flagfaulttolerance} are not easily described in this language.
Since this section is only concerned with phenomenological noise models, we do not take into account errors that happen during the execution of the measurement circuits and error propagation in these circuits.

To define a quantum data-syndrome code, we first pick an underlying stabilizer code with stabilizer group $\stabs \subseteq \pauli^n$.
In each round, instead of just a set of generators, a larger set of  redundant stabilizers $g_1,\cdots,g_m \in \stabs$ is measured.
The simplest and most common case is to simply repeat the measurements of the generators.
More generally, the redundant stabilizers can be chosen according to a classical code, as described in \cite{ashikhmin_quantumdatasyndromecodes}.
An error can then be described by a data error $e_d \in \pauli^n$ and a measurement error $e_m \in \Ftwo^m$, i.e.\ $e_m[i] = 1$ if the measurement of $g_i$ returned the wrong outcome and $e_m[i] = 0$ otherwise.
The measurements of the generators, including measurement errors, can be described by the extended parity check matrix
\begin{equation}
H = \begin{bmatrix}
G & I_m
\end{bmatrix} \, ,
\end{equation}
where the rows of $G$ are the original stabilizers.
That is, each generator $g_i \in \pauli$ is extended to an element $f_i = (g_i, \hat i) \in \paulids$, where $\hat i$ is the $i$-th standard basis vector.
Then, the outcome of the measurement of $f_i$ if an error $e = (e_d,e_m) \in \paulids = \pauli^n \times \Ftwo^m$ occurred is exactly given by $\scprod{f_i}{e}$, using the bicharacter of $\paulids \coloneqq \pauli^n \times \Ftwo^m$ (\cref{eq:BicharacterPauliDS}).
The set of measurements we have access to is thus the group generated by the extended stabilizers $f_i$,
\begin{equation}
\meas \coloneqq \langle f_1,\dots,f_m \rangle \, .
\end{equation}
As always, the collection of measurement outcomes for all $s \in \meas$ is called the \emph{syndrome}, and it can be obtained by measuring the generators $f_i$ of $\meas$.
The undetectable errors $\undet = \meas^\perp$ are exactly those that result in a trivial syndrome.

Similarly to subsystem codes, the concepts of measurements and stabilizers do not coincide.
In fact, the measurements do not corresponds to undetectable errors, $\meas \not \subseteq \undet$.
Since errors differing by elements of $\meas$ do not have the same syndrome
it follows that, in particular, they cannot be considered logically equivalent.
Instead, logical equivalence is still described by the stabilizer group $\stabs \subseteq \pauli^n$ of the underlying code, which we view as a subgroup of $\paulids$.
The \emph{logical operators} $\logops$ are those operators that map the codespace of the underlying quantum code to itself, i.e. $\logops \coloneqq \stabs^\perp$, where the annihilator is in $\paulids$, not just in $\pauli^n$.
These groups then fulfill the dual inclusion relations $\stabs \subseteq \undet$ and $\meas \subseteq \logops$.

Motivated by the discussion above, the \emph{distance} of a data-syndrome code is defined as the minimal weight of an element of $\undet \setminus \stabs$ \cite{ashikhmin_quantumdatasyndromecodes}.
Remember that a region $R \subseteq [n+m]$ is correctable if there is no element of $\undet \setminus \stabs$ that is supported on $R$ (\cref{def:CorrectableRegion}).
Thus, as expected, if $|R| < d$, $R$ is correctable.

Consider a phenomenological noise model where in each round a new error $e = (e_d,e_m)$ occurs according to a distribution $P$.
If we were to always reset to the ground state of our code between two rounds of measurements, we could now directly apply \cref{thm:Main}, setting $\gauge = \stabs$.
However, in a more realistic setting we want to preserve the information between rounds and thus our measurements will act on the accumulated data error in each round and not just on the new error.
Similar to the previous section, we can remedy this by considering the syndrome relative to the previous one.
However, this will effectively propagate measurement errors between rounds.
If $a_d$ is the accumulated data error in a given round, $a_m$ the measurement errors in that round, and $e = (e_d,e_m)$ is the new error occurring in the next round, then the product of outcomes for the measurement $s \in \meas$ is given by
\begin{equation}
\scprod{s}{(a_d,a_m)}\scprod{s}{(a_d,0)(e_d,e_m)} = \scprod{s}{(e_d,e_ma_m)} \, .
\end{equation}
Thus, effectively we measure the new data error $e_d$ and the combined measurement error $a_me_m$ from both rounds.
Since the measurement errors are assumed to be independent between rounds, the distribution $\tilde P$ of $(e_d,a_me_m)$ factorizes in the same way as $P$, but the strength of measurement errors is increased.
Here, it is important that we divide the measurements into disjoint pairs of consecutive rounds such that the measurement errors are also independent between each pair.
By \cref{thm:Main}, as long as the original noise is correctable, we can then estimate the adjusted distribution from the syndrome measurements, up to logical equivalence.

Depending on how exactly error correction is performed, $\tilde P$ might be the most relevant distribution.
If we are instead interested in the original error distribution $P$, we can also obtain this by post-processing as follows.
Denote as $P_m$ the marginal distribution of $P$ on the measurement errors, i.e.
\begin{equation}
P_m(e_m) = \sum_{e_d \in \pauli^n} P(e_d,e_m) \, .
\end{equation}
We can write this as $P_m = (P \conv \Phi_{\pauli^n}) \cdot \Phi_{\Ftwo^m}$, where we view $\pauli^n$ and $\Ftwo^m$ as subgroups of $\paulids$.
Since $(\pauli^n)^\perp = \Ftwo^m$, we have by \cref{lem:IndicatorDuality},
\begin{equation}
\begin{split} 
E_m &\coloneqq \fourier{P_m} = \fourier{P \conv \Phi_{\pauli^n}} \conv \fourier{\Phi_{\Ftwo^m}} \nonumber \\
& = \bigl(E \cdot \abs{\paulids}U_{\Ftwo^m}\bigr) \conv \abs{\paulids}U_{\pauli^n} \nonumber \\
& = \frac{\abs{\paulids}^2}{\abs{\pauli^n}\abs{\Ftwo^m}} (E \cdot \Phi_{\Ftwo^m})\conv \Phi_{\pauli^n}
= \abs{\paulids} (E \cdot \Phi_{\Ftwo^m})\conv \Phi_{\pauli^n} \nonumber \, .
\end{split}
\end{equation}
Explicitly, this means that $E_m(e_d,e_m) = \abs{\paulids} \,E(0,e_m)$.
Since the measurement errors are independent between rounds, the adjusted distribution is given by $\tilde P = P \conv P_m$.
Thus, the moments of the adjusted distribution are $\tilde E \coloneqq \fourier{\tilde P} = E \cdot E_m$.
We can obtain these up to logical equivalence, i.e.\ we obtain $\tilde E \cdot \Phi_{\logops}$, and are interested in the moments $E \cdot \Phi_{\logops}$ of the original logical channel.
Since $\gauge = \stabs \subseteq \pauli^n$, we have $\Ftwo^m \subseteq \logops = \gauge^\perp$.
Thus, in particular we have access to the moment $\tilde E(0,e_m)$ for any $e_m \in \Ftwo^m$, and by the  above discussion we have $\tilde E(0,e_m) = E(0,e_m)E_m(0,e_m) = \frac{1}{\abs{\paulids}}E_m(0,e_m)^2$.
Thus we can obtain the original moment for each $l = (l_d,l_m) \in \logops$ from the adjusted moments as follows:
\begin{equation}
E(l_d,l_m) = \frac{\tilde{E}(l_d,l_m)}{E_m(l_d,l_m)} = \frac{1}{\sqrt{|\paulids|}}\frac{\tilde E(l_d,l_m)}{\sqrt{\tilde E(0,l_m)}} \, .
\end{equation}

All in all we obtain the following corollary to \cref{thm:Main}.

\begin{corollary}
Phenomenological data and measurement noise with error rates smaller than $\frac{1}{2}$ can be estimated from the measurements of a quantum data-syndrome code up to logical equivalence if the noise is not correlated over more than half the distance of the code.
\end{corollary}


\begin{acronym}[LDPC]
\acro{ACES}{averaged circuit eigenvalue sampling}
\acro{AGF}{average gate fidelity}
\acro{AP}{Arbeitspaket}

\acro{BOG}{binned outcome generation}

\acro{CP}{completely positive}
\acro{CPT}{completely positive and trace preserving}
\acro{cs}{computer science}
\acro{CS}{compressed sensing} 

\acro{DAQC}{digital-analog quantum computing}
\acro{DD}{dynamical decoupling}
\acro{DFE}{direct fidelity estimation} 
\acro{DFT}{discrete Fourier transform}
\acro{DM}{dark matter}

\acro{FFT}{fast Fourier transform}

\acro{GST}{gate set tomography}
\acro{GTM}{gate-independent, time-stationary, Markovian}
\acro{GUE}{Gaussian unitary ensemble}

\acro{HOG}{heavy outcome generation}

\acro{irrep}{irreducible representation}

\acro{LDPC}{low density partity check}
\acro{LP}{linear program}

\acro{MAGIC}{magnetic gradient induced coupling}
\acro{MBL}{many-body localization}
\acro{MIP}{mixed integer program}
\acro{ML}{machine learning}
\acro{MLE}{maximum likelihood estimation}
\acro{MPO}{matrix product operator}
\acro{MPS}{matrix product state}
\acro{MS}{M{\o}lmer-S{\o}rensen}
\acro{MUBs}{mutually unbiased bases} 
\acro{mw}{micro wave}

\acro{NISQ}{noisy and intermediate scale quantum}

\acro{ONB}{orthonormal basis}
\acroplural{ONB}[ONBs]{orthonormal bases}

\acro{POVM}{positive operator valued measure}
\acro{PSD}{positive-semidefinite}
\acro{PSR}{parameter shift rule}
\acro{PVM}{projector-valued measure}

\acro{QAOA}{quantum approximate optimization algorithm}
\acro{QC}{quantum computation}
\acro{QEC}{quantum error correction}
\acro{QFT}{quantum Fourier transform}
\acro{QM}{quantum mechanics}
\acro{QML}{quantum machine learning}
\acro{QMT}{measurement tomography}
\acro{QPT}{quantum process tomography}
\acro{QPU}{quantum processing unit}
\acro{QUBO}{quadratic binary optimization}

\acro{RB}{randomized benchmarking}
\acro{RBM}{restricted Boltzmann machine}
\acro{RDM}{reduced density matrix}
\acro{rf}{radio frequency}
\acro{RIC}{restricted isometry constant}
\acro{RIP}{restricted isometry property}

\acro{SDP}{semidefinite program}
\acro{SFE}{shadow fidelity estimation}
\acro{SIC}{symmetric, informationally complete}
\acro{SPAM}{state preparation and measurement}

\acro{TT}{tensor train}
\acro{TM}{Turing machine}
\acro{TV}{total variation}

\acro{VQA}{variational quantum algorithm}

\acro{VQE}{variational quantum eigensolver}

\acro{XEB}{cross-entropy benchmarking}

\end{acronym}

\bibliographystyle{./myapsrev4-2}
\providecommand{\njp}{New J.\ Phys.}
\bibliography{bibliography,mk}

\end{document}